\documentclass[11pt]{article}
\usepackage[utf8]{inputenc}
\usepackage[parfill]{parskip}
\usepackage{setspace}
\usepackage[unicode,psdextra,colorlinks=true,citecolor=blue,bookmarksnumbered=true,linkcolor=green!40!black]{hyperref}

\usepackage{fullpage}
\usepackage{amssymb,amsmath,amsopn,amsthm, mathtools}
\usepackage{graphicx}
\usepackage{xcolor}
\usepackage{cleveref}

\newtheorem{theorem}{Theorem}[section]

\newtheorem{definition}[theorem]{Definition}

\newtheorem{claim}[theorem]{Claim}

\newtheorem{question}[theorem]{Open Question}

\newcommand{\ignore}[1]{}

\newcommand{\col}{\mathsf{col}}
\newcommand{\mult}{\mathsf{mult}}

\title{Combinatorial lower bounds for $3$-query LDCs}

\author{Arnab Bhattacharyya \thanks{National University of Singapore, Singapore. Email: \texttt{arnabb@nus.edu.sg}. This work was partially supported by NUS Startup Grant R-252-000-A33-133. } \and L.~Sunil Chandran \thanks{Indian Institute of Science, Bangalore. E-mail: \texttt{sunil@iisc.ac.in} } \and Suprovat Ghoshal\thanks{Indian Institute of Science, Bangalore. E-mail: \texttt{suprovat@iisc.ac.in}} }

\begin{document}

\maketitle
\thispagestyle{empty}

\begin{abstract}
	A code is called a {\em $q$-query locally decodable code} (LDC)  if there is a randomized decoding algorithm that, given an index $i$ and a received word $w$ close to an encoding of a message $x$, outputs $x_i$ by querying only at most $q$ coordinates of $w$. Understanding the tradeoffs between the dimension, length and query complexity of LDCs is a fascinating and unresolved research challenge. In particular, for $3$-query binary LDC's of dimension $k$ and length $n$, the best known bounds are: $2^{k^{o(1)}} \geq n \geq \tilde{\Omega}(k^2)$.
	
	In this work, we take a second look at binary $3$-query LDCs. We investigate a class of 3-uniform hypergraphs that are equivalent to {\em strong} binary 3-query LDCs. We prove an upper bound on the number of edges in these hypergraphs, reproducing the known lower bound of $\tilde{\Omega}(k^2)$ for the length of strong $3$-query LDCs.  In contrast to previous work, our techniques are purely combinatorial and do not rely on a direct reduction to $2$-query LDCs, opening up a potentially different approach to analyzing 3-query LDCs.
\end{abstract}

\section{Introduction}

A code $\mathcal{C}$ is said to be a {\em $q$-query locally decodable code (LDC)} if it is possible to recover any symbol $x_i$ of a message $x$ by querying $\mathcal{C}(x)$ on at most $q$ locations, such that even if a constant fraction of $\mathcal{C}(x)$ is corrupted, the decoder returns $x_i$ with high probability. LDCs already appeared in the PCP literature (e.g., implicitly in \cite{BFLS}) but they were first explicitly formulated by Katz and Trevisan in \cite{KT00}. 
LDCs have attracted attention not only because of their immediate relevance to data transmission and data storage but also because of their surprising connections to complexity theory and cryptography (\cite{CGKS98, STV01, DS07,KS09}). In more recent years, the analysis of LDCs has led to a greater understanding of basic problems in incidence geometry, the construction of design matrices and the theory of matrix scaling, e.g. \cite{BDWY11, DSW12, DSW14}. 

Although LDCs have been studied now for two decades, some basic questions remain stubbornly open. In particular, we have the following open question for 3-query LDCs:
\begin{question}
	What is the length of the shortest $3$-query LDC that can encode all $k$-bit binary messages?
\end{question}
A wild variety of techniques have been used to study the problem. For a while, it was believed that the length  $n$ should be exponential in $k$ for $3$-query LDCs (indeed, for any constant number of queries). This belief was shattered by a breakthrough work of Yekhanin that designed $3$-query LDCs of length subexponential in $k$ (conditional on some number-theoretic conjectures). Subsequent work (\cite{Efr12, DGY11} reformulated the construction in terms of {\em matching vector codes} and established an unconditional upper bound of $\exp\left(\exp(O(\sqrt{\log k \log \log k}))\right) = \exp(k^{o(1)})$ on the length.  

As for lower bounds on the length of $3$-query LDCs, which is the focus of this work, Katz and Trevisan~\cite{KT00} first gave a super linear lower bound of $\Omega(k^{3/2})$, which was then improved to $\Omega\left(k^2/(\log k)^2\right)$ by Kerenedis and de Wolf~\cite{KdW03} using quantum information theoretic techniques. The current state-of-the-art is due to Woodruff \cite{Woo07} from over a decade ago where he showed that $n \geq {\Omega}(k^2/\log k)$. 

Given the state of affairs, it is natural to try to prove lower bounds for stronger variants\footnote{For instance, Woodruff in \cite{Woo12} gave an $\Omega(k^2)$ lower bound for the special case of {\em linear} $3$-query LDCs.
} of LDCs where the task should be easier. In this work, we study a restricted form of LDCs which seem to capture most of the challenges associated with general LDCs. 
\begin{definition}\label{def:strldc}
	For a given $\delta \in (0,1)$, 
	a code $\mathcal{C}: \{\pm 1\}^k \to \{\pm 1\}^n$ is
	a {\em $(3,\delta)$-strong LDC} if for every $i \in [k]$, there exists a set $M_i$ of $\geq \delta n$ disjoint triples in ${[n] \choose 3}$ such that for every $x \in \{\pm 1\}^k$ and for every triple $(j_1, j_2, j_3) \in M_i$, $x_i = \mathcal{C}(x)_{j_1}\cdot \mathcal{C}(x)_{j_2}\cdot \mathcal{C}(x)_{j_3}$. Moreover, if $i \neq i'$, a triple in $M_i$ intersects a triple in $M_{i'}$ in at most 1 coordinate. 
\end{definition}
Known constructions of 3-query LDCs are strong. Conceptually, the main\footnote{The decoding scheme of taking the product (xor) of the codeword bits is without loss of generality (see \cite{Woo07}). The additional condition in \Cref{def:strldc} about triples in different matchings intersecting only at single coordinates is made for technical convenience and should be avoidable.} restriction that the above definition makes is that each triple in the matching $M_i$ successfully decodes $x_i$ for {\em every} $x$. On the other hand, Katz and Trevisan \cite{KT00} show that general LDCs yield matchings $M_1, \dots, M_k$ such that each triple in the matching $M_i$ sucessfully decodes $x_i$ for most (not all) $x$'s.

We show a combinatorial proof of the known $\Omega(k^2/\log k)$ lower bound for  the length of code words of $3$-query strong LDCs.  Here is the main theorem stating the lower bound.

\begin{theorem}				\label{thm:main}
	Let $C:\{\pm 1\}^k \mapsto \{\pm 1\}^n$ be a $(3,\delta)$-strong LDC. Then, $n \ge \Omega_\delta(k^2/\log k)$.
\end{theorem}

\subsection{Proof Overview}
As we already noted, Theorem \ref{thm:main} follows from \cite{Woo07}. Of more interest is our proof technique. Woodruff's lower bound reduces $3$-query LDCs to $2$-query LDCs and applies known analytic proofs giving tight bounds for $2$-query LDCs \cite{KdW03}. On the other hand, our proof is purely combinatorial and does not seem to be a reduction to 2 queries. 

Our starting point is the observation that strong 3-query LDCs are equivalent to {\em even-colored 3-uniform hypergraphs}. A 3-uniform hypergraph is called {\em linear} if any two edges intersect in at most one vertex.
\begin{definition}\label{def:linh}
	An {\em $(n,k,\delta)$-even-colored 3-uniform hypergraph} is a linear edge-colored hypergraph $H$ on $n$ vertices with each edge having a color in $\{1,\dots, k\}$ such that:
	\begin{itemize}
		\item[(i)] For each $i \in [k]$, the edges of color $i$ form a matching of size at least $\delta n$, and
		\item[(ii)] If $H'$ is a subgraph of $H$ such that every vertex has even degree in $H'$, then there are an even number of edges of each color in $H'$. 
	\end{itemize}
\end{definition}

Given a $(3,\delta)$-strong LDCs, define the hypergraph $H$ which is the union of the matchings $M_1,\dots, M_k$ given by Definition \ref{def:strldc}, and let the color of an edge be the matching it comes from. Then, it is easy to check that both conditions (i) and (ii) are met (see \Cref{cl:even-subgraph-color}). The correspondence naturally goes in the other direction too, although this is not needed in the present work.

We prove an upper bound $k^2 \leq O_\delta(n \log n)$  for $(n,k,\delta)$-even-colored 3-uniform hypergraphs, proving our main theorem. To motivate our proof, let us sketch the corresponding argument for 2-query LDCs (which is also new to the best of our knowledge). Suppose we have a (2-uniform) graph which is the union of $k$ matchings, with edges from the $i$'th matching having color $i$. Analogously to condition (ii) of \Cref{def:linh}, also suppose that any cycle contains an even number of edges of each color. Then, we prove that the number of vertices $n$ is at least $\exp(k)$.  For simplicity, suppose the matchings of each color are perfect. Our argument is through coding (ironically!). Fix an arbitrary vertex $s$. For any vertex $v \neq s$, let its {\em signature} $S(v)$ be defined as $(n_1, \dots, n_k)$ where $n_i$ is the parity of the number of edges of color $i$ on a path $P$ from $s$ to $v$. We claim that $S(v)$ does not depend on the path chosen. This is because if two paths from $s$ to $v$ gave different signatures, this would yield a cycle in which some color occurred an odd number of times.  On the other hand, there are at least $2^k$ different signatures because for any signature $(s_1, \dots, s_k) \in \{0,1\}^k$, there is a path from $s$ with exactly $s_i$ edges of color $i$ (since the matchings are perfect). Hence, the number of vertices is at least $\exp(k)$.

\begin{figure}[b]
	\begin{center}
		\includegraphics[width=1.5in]{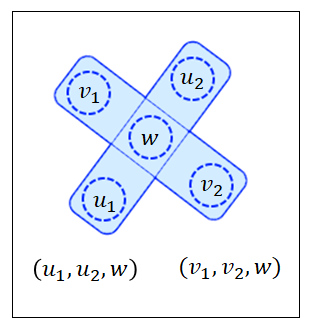}
	\end{center}
	\caption{A cherry formed from the edges $\{u_1, w, u_2\}$ and $\{v_1, w, v_2\}$ intersecting at $w$.}
	\label{fig:cherry}
\end{figure}

Our proof for 3-uniform hypergraphs is in a similar spirit. Instead of a path to define a signature, we use a sequence of {\em cherries}, borrowing an idea from \cite{dellamonica2012even}. A cherry is a pair of hyperedges which uniquely intersect at a hyperedge; see \Cref{fig:cherry}. We observe that if the number of edges is sufficiently large, then there are many cherries. We then use this structure to show that there are even subgraphs (i.e., subgraphs in which all vertices have even degrees) which have an odd number of edges of some color. Namely, we construct a `cycle of cherries' in which we know there is a color that appears on a unique edge, yielding the contradiction. More details follow.

Formally, given an $(n,k,\delta)$-even colored hypergraph $H$ which is a union of matchings $M_1, \dots, M_k$, define the {\em signature graph} $G = (V',E')$ as follows. The vertex set $V' = V \times V$, and there is an edge in $G$ between the vertices $(u_1,v_1)$ and $(u_2,v_2) \in V'$ whenever there exists a $w \in V$ such that $\{u_1,u_2,w\}$ and $\{v_1,v_2,w\}$ are hyperedges forming a cherry in $H$. Moreover, such an edge is labeled by the pair of colors $\{i_1, i_2\}$ if $\{u_1, u_2, w\} \in M_{i_1}$ and $\{v_1, v_2, w\} \in M_{i_2}$. The signature graph enjoys the following useful structural property: 
\begin{itemize}
	\item[]For any vertex $x$ in $G$ and for any color $i \in [k]$, there are at most $4$ edges incident to $x$ that have $i$ in their label.\hspace*{\fill}($\star$)
\end{itemize}
The proof of ($\star$) follows from the definition of $G$ in terms of cherries (see \Cref{cl:neighborhood-color-bound}) .

For the sake of contradiction, assume that $k \geq \sqrt{C n \log n}$ for some large constant $C$. This, along with a standard averaging argument, implies that there exists a large subgraph $G'$ of the signature graph with minimum degree at least $k^2n/4n^2 \geq (C/4)\log n$. Now fixing an arbitrary vertex $r \in V(G')$ as root, we iteratively grow a sequence of trees $T_1,T_2,\ldots,T_{\ell+1}$ using edges in $G'$, while maintaining the following {\em rainbow} condition: For any vertex $x \in V(T_i)$, no color appears more than once among the colors labeling the unique path from $r$ to $x$ in $T_i$. 

We explain how to construct $T_{i+1}$ from $T_i$ so that the above rainbow condition is met. Let $L_i$ denote the leaves of the tree $T_i$, and let $N_i$ denote the neighbors $y$ of vertices $x \in L_i$ so that the colors labeling the edge $(x,y)$ do not occur on the path from the root $r$ to $x$ in $T_i$. A short argument (\Cref{cl:property-of-property}) allows us to deduce that $N_i$ must be disjoint from $V(T_i)$, because otherwise, condition (ii) of \Cref{def:linh} is violated. Hence, the next tree $T_{i+1}$ can be built by letting $L_{i+1}=N_i$ and adding one edge from each vertex in $L_{i+1}$ to a vertex in $L_i$.

We continue this process until $|L_{\ell}|< 2|L_{\ell-1}|$ for some iteration $\ell$. We now sketch how to arrive at a contradiction. From the stopping criteria, we know that for every $i< \ell-1$ we have $|L_{i+1}| \geq 2|L_i| \geq 2^{i+1}$, and therefore the depth of the tree is at most $\ell  = O(\log n)$. Therefore, for any $x \in L_{\ell-1}$, the number of colors labeling the path from $r$ to $x$ is at most $O(\log n)$. From property ($\star$), we get that there are at least $(C/4) \log n - O(\log n) = C' \log n$ neighbors of $x$ that are in $L_{\ell}$ (for some other constant $C'$).  Since $|L_{\ell}| < 2 |L_{\ell-1}|$, there exists a vertex $w \in L_{\ell}$ with at least $(C'/2) \log n$ neighbors in $L_{\ell-1}$. Again, invoking property ($\star$), for $C'$ large enough, there will be a neighbor $w' \in L_{\ell-1}$ such that that the colors labeling $(w,w')$ do not appear among the $O(\log n)$ labels of the path from $r$ to $w$. From here, we can conclude that the unique path between $w$ and $w'$ in $T_{\ell}$ along with the edge $(w,w')$ forms a cycle in $G$ in which some color appears exactly once. This structure corresponds to a subgraph in $H$ that violates condition (ii) of \Cref{def:linh}.

In the rest of the paper, we present the argument formally with all the details. It is unclear currently how to extend the analysis to $q$-query LDCs or how to improve the analysis for $3$-query LDCs. But we remain hopeful that by looking at more intricate combinatorial structures than cherries, we can make some progress.

\section{Preliminaries}

In this section and later, we do not invoke the notion of  even-colored subgraphs, and we define objects directly in reference to strong 3-query LDCs.

Given a $(3,\gamma)$-strong LDC, we define the \emph{recovery hypergraph} $H$, where
$V(H) =[n]$ and $E(H):= \cup_{i \in [k]} M_i$ to be the $3$-uniform hypergraph which is the union of matchings $M_1,M_2,\ldots,M_k$. 
For any edge $e \in E(H)$, we say that the color of the hyperedge $e$ is $i$ if $e$ belongs to matching $M_i$. We use the notation $\col(e)$ to denote the  color of the
hyperedge $e$.  
We additionally assume that $H$ is linear, i.e.  no two hyperedges of $H$  intersect in more than one element. 

Let $L$ be a hypergraph. Then we define an {\em augmentation} $L'$ of $L$ as follows: $V(L') = V(L)$ and $E(L')$  is  a multiset where each member $e 
\in E(L')$ also belongs to $E(L)$ but can possibly have a higher multiplicity than the multiplicity of $e$ in  $E(L)$.   With respect to  a hypergraph $L$  where a hyperedge $e$ is allowed to have 
multiplicity greater than 1,  we denote by $\mult_L(e)$ the multiplicity of $e$  in $E(L)$. We may drop the subscript $L$, if the hypergraph under consideration is  clear from the 
context. Also for $v \in V(L),   \deg_L(v) := \sum_{e \in E(L): v \in e} \mult_L(e)$.  If for all $ v \in V(L), \deg_L(v)$ is even, then  $L$ is called an \emph {even  hypergraph}. 
We use $\col(L)$ to  denote the multiset of colors associated with edges in $E(L)$. 

If $L$ is an augmentation of  the recovery hypergraph $H$, for $1 \le i \le k$,  we define 
\[
\lambda_L(i) := \sum_{e \in E(L): \col(e) = i} \mult_L(e).
\]

~~~

\begin {claim} \label {cl:even-subgraph-color} 

Let $L$ be an augmentation of the recovery hypergraph $H$.  If $L$ is even, then  for $1 \le i \le k$, $\lambda_L (i)$ is even.

\end {claim}

\begin {proof}

Suppose for contradiction that there exists $i$, $1 \le i \le k$ such that  $\lambda_L (i)$ is odd.  Recall that the indices of the code word bits correspond to the 
vertices of the recovery hypergraph.
Let us assume that $y=(y_1, y_2, \ldots, y_n)$ is the code word of a message $(x_1, x_2, \ldots, x_k)$, where $x_i = -1$ and $x_j = 1$ for $j \ne i$.  For an edge
$e =\{a,b,c\}$ of the recovery hypergraph,  $Y^e = y_a.y_b.y_c$.  Now it is clear that  $\prod_{ e \in E(L)} Y^e = 1$ , since $L$ is an even augmentation of $H$.  

On the other hand, by definition of the recovery hypergraph, if $\col(e) = t$, then $Y^e = x_t$ for $1 \le t \le k$. Therefore $\prod_{e \in E(L)} Y^e =  \prod_{1 \le t \le k}  x_t^{\lambda_L(t)}$.
Clearly since for the selected message   $x_j =1$ for $j \ne i$, we infer that  $\prod_{e \in E(L)} Y^e = x_i^{\lambda_L(i)} = -1$, if $\lambda_L(i)$ is odd. This is a contradiction. 
We conclude that for $1 \le i \le k$, $\lambda_L (i)$ is even.

\end {proof}

\paragraph*{The Signature Graph.}   We define a graph called the {\em signature graph} $G$ as follows:  $V (G) = \{ (u,v) : u, v \in V(H); u \ne v \}$ and an edge exists between 
two vertices $(u_1,v_1)$ and $(u_2,v_2)$ of $G$   if and only if  $\{u_1,  v_1 \} \cap \{u_2, v_2 \} = \emptyset $ and  there exists a vertex $w \in V(H)$ such 
that $\{u_1,u_2,w\},\{v_1,v_2,w\} \in E(H)$. 
Note that  since the recovery  hypergraph $H$ is linear, if there exists an edge between two vertices $(u_1,v_1)$ and $(u_2,v_2)$, there is a unique vertex $w$ 
such that  $\{u_1,u_2,w\}, \{v_1,v_2, w\} \in E(H)$.  
We may say that the vertex $w$  \emph {causes}  the edge $((u_1,v_1), (u_2,v_2))$. Given an edge $e \in E(G)$, we define $T(e) = (\{u_1,u_2,w\}, \{v_1,v_2, w\})$ if $w$ causes the edge $e$.  We may abuse
the notation and use $T(e)$ to denote the corresponding unordered set.  
We define   $\col({e}) = \{i_1,i_2\}$ if $(u_1,u_2,w) \in M_{i_1}$ and $(v_1,v_2,w) \in M_{i_2}$. Note that $i_1 \ne i_2$ since $w$ cannot be in two different
edges of the same matching.

\begin{claim}					\label{cl:edges-bound}
	The number of edges in the signature graph $G$ is at least $12 \gamma^2 nk^2$. 
\end{claim}
\begin{proof}
	Since each matching is of size at least $\gamma n$, the number of hyperedges $m$ in $H$ is at least $\gamma nk$.    It follows that  $\sum_{v \in V(H)} \deg(v) = 3m \ge  3 \gamma nk$. 
	For any vertex $w \in V(H)$  consider a pair of incident edges, say $\{u_1,u_2,w\}$ and $\{v_1,v_2,w\}$. Since $H$ is linear, $\{u_1,u_2\}  \cap \{v_1,v_2\} = \emptyset$. 
	It is easy to see that based on  this pair of incident edges, $w$  can \emph {cause}  4  distinct edges of the signature graph $G$.  Therefore the  vertex $w$  \emph {causes} 
	$4 { d(v) \choose 2}$ distinct edges of  $G$,
	where $d(v) =\deg(v)$.
	As we have mentioned earlier,  two different vertices $w$ and $w'$ cannot cause the same edge in $G$. Therefore  $|E (G) | \ge  4 \sum_{v \in V(H)} {d(v) \choose 2} = 2 \sum_{v \in V(H)} (d(v)^2 
	-d(v) )$.  Recall that $\sum_{v \in V(H) } d(v) = 3m$.  Using the Cauchy-Schwarz inequality,
	\footnote {For vectors $X,Y$, the Cauchy-Schwarz inequality states that $\|X\|\cdot\|Y\| \ge X^\top Y$.  Now take $X = (d(v_1), d(v_2), \ldots, d(v_n))^\top$, where
		$V(H) =\{v_1,v_2,\ldots,v_n\}$ and $Y=(1,1, \ldots,1)^T$ to get the
		required lower bound. }
	we get $\sum_{v \in V(H) } (d(v)^2) \ge \frac {9m^2}{n}   $.
	It follows that $|E(G)| \ge  \frac {6m}{n} (3m-n) \ge \frac {12m^2}{n} \ge 12 \gamma^2 n k^2 $.    Here we have used $m \ge \gamma n k$ and  $m \ge  n$, since $\gamma k \ge 1$. 
	
\end{proof}

For a subgraph $J$ of the signature graph $G$, we define $H_J$ to be the augmentation of $H$, with $V(H_J) = V(H)$ and $E(H_J) = \cup_{e \in E(J)} T(e) $. Note that when we take
the union here, we retain multiple copies of a hyperedge if that hyperedge appears in multiple sets $T(e)$  taking part in the union operation. Thus $E(H_J)$ is by definition a multi-set. 
We extend  some of the notation used for hypergraphs to subgraphs of signature graphs also in the following way:   We use $\col(J)$ to denote the multiset $\col(H_J)$. 
A hypergraph $H'$ is {\em rainbow colored} with respect to an edge coloring if there exist no two  hyperedges having the same color. (In particular, there will not be any hyperedge with multiplicity greater than 1.)
A subgraph $J$ of the signature graph is rainbow colored, if $H_J$ is rainbow colored.  We may also say $J$ (or $H$) is \emph {rainbow}, shortening the phrase  \emph {rainbow colored}. 

~~~~~

\begin {claim} \label  {cl:even-subgraph}

Let $J$ be an even subgraph of the signature graph $G$, i.e.  $ \forall v \in V(J)$, $\deg(v)$ is even.  Then $H_J$ is an \emph {even} augmentation of $H$.

\end {claim}

\begin {proof}

Recall that each edge $e = ( (u_1,v_1), (u_2,v_2)) \in E(G)$  corresponds to exactly  $2$ edges   in $E(H_J)$, namely the two edges of $T(e) = $ $ \{ \{u_1,u_2,w\},$ 
$ \{v_1,v_2, w\} \}$, where $w$ is
the  unique vertex which \emph {caused}  the edge $e$. We say that $w$ appears in the role of an intermediate vertex and $u_1,v_1,u_2$ and $v_2$ appear in the role of signature vertices in $T(e)$. 
It is easy to see that since in $T(e)$ itself the $degree(w)$ is even, each vertex plays the role of an intermediate vertex an even number of times.   Noting that  in $T(e)$ each vertex
appears in the role of a signature vertex exactly once,  it is easy to see that if $x =(u,v)$ is a vertex of $J$, then  $u$ (also $v$) plays the role  of a signature vertex in
$\cup_{e \in E_J(x)} T(e)$   (where $E_J(x)$ denotes the set  of edges incident on $x$ in $J$) exactly $\deg_{J} (x)$ times. 
Since $\deg_J(x)$ is even, it follows that $\deg_{H_J} (u)$ and $\deg_{H_J} (v)$ are even numbers. 
\end {proof}

\begin {claim}  \label {cl:neighborhood-color-bound}

Let  $x$ be a vertex of the signature graph $G$  and let $E(x)$ be the set of edges incident on $x$ in $G$.  Let ${\cal C} \subseteq [k]$ be a subset of  colors.  Let 
$E (x,{\cal C}) =  \{e \in E(x) :    \col(e) \cap {\cal C} \ne   
\emptyset  \}   $. 
Then $|E(x,{\cal C})| \le  4 | {\cal C} | $.

\end {claim}

\begin {proof}

Let $x = (u,v)$.  For an edge $e \in E(x)$, $T(e)$ contains 2 hyperedges, exactly one of which contains $u$ and the other one contains $v$:
Let us denote by $T(e)_1$ and $T(e)_2$ the hyperedges in $T(e)$ that contain $u$ and $v$ respectively.  For $1 \le i \le k$, 
$E(x, {\cal C} )  = \cup_{i \in {\cal C} } E_1^i \cup E_2^i$, where $E_j^i = \{ e \in E(x) : \col(T(e)_j) = i\}$ for
$j =1,2$.  First we will show that $|E_1^i| \le 2$.  To see this, note that if  $\col(T(e)_1) = i$, then    $T(e)_1 \in M_i$ and $T(e)_1$ contains
$u$ as mentioned earlier. There is
a unique hyperedge with these properties since $M_i$ is a matching.  Let $T(e)_1 = (u,a,b)$. Then either $a$ or $b$  could have \emph {caused} the edge $e$.
If $a$ caused the edge $e$, then $T(e)_2$ contains  both $v$ and $a$, and there is a unique edge in $E(H)$  that is a superset of  $\{v,a\}$ since $H$ is linear.
Similarly if $b$ caused $e$, then $T(e)_2$ is uniquely determined, since it should contain both $v$ and $b$.  It follows that
$|E_1^i | \le 2$ for all $i$, $1 \le i \le k$.  A similar argument shows that  $|E_2^i| \le 2$.  
It follows that $|E(x, {\cal C} ) | \le 4 | {\cal C} |$.

\end {proof}

Now we are ready to prove Theorem \ref{thm:main}.

\section{Proof of Theorem \ref{thm:main}}

For contradiction, we shall assume that $k^2 >  C n\log n$ where $C = C(\gamma)$ is a sufficiently large constant. We can then lower bound the
average degree of $G$ as follows. From claim \ref{cl:edges-bound}, we know that $|E(G)| \geq 12 \gamma^2 nk^2$. On the other hand, the number of vertices in $G$ is at most $n^2$. Therefore, for $k^2 \ge C n\log n$, for $C$ large enough,  the average degree of $G$
is at least $(2C') \log n$ so that we can find a subgraph $G' \subseteq  G$ with minimum degree at least $C' \log n$, where $C'$ is a sufficiently large constant. 

Now we fix a vertex $r \in V(G')$, and we grow a rainbow tree rooted at $r$ in $G'$ level by level as follows. 
Let $T_0$ be the tree consisting only of the root $r$ and  $T_1$ be the tree consisting  of $r$ and all its neighbors.  At the $i$th stage, 
we will have a tree $T_i$ where $V(T_i) = \dot \cup_{i=0}^{i} L_i$ where $L_i$ is the set of  vertices in level $i$. Note that $V_0 = L_0 = \{r\}$ and
$T_1$ consists of 2 levels, $L_0 = \{r\}$ and $L_1 = N_{G'} (r)$.    For two vertices $x$ and $y$, the unique path in $T_i$ from $x$ to $y$ will be denoted by $P(x,y)$.  

Moreover at the $i$th stage, we will make sure that the tree $T_i$ satisfies the following property:

\begin {equation}  \label  {eq:property-of-tree}
\mbox {For any vertex }   x \in V(T_i),   P(r,x)  \mbox { is a rainbow path.} 
\end {equation}

\begin {claim} \label{cl:property-of-property}
If $T_i$ satisfies property (\ref {eq:property-of-tree}), and if $e=(x,y)$ is an edge of $G'$ such that $x \in L_i$ and $y \in L_j$, where $i \ge j$, 
then  $\col(e) \cap \col(P(r,x)) \ne \emptyset$. 
\end {claim}

\begin {proof}

Suppose not. Let $z$ be the least common ancestor of $x$ and $y$ in $T_i$. Since $P(z,x) \subseteq P(r,x)$ is a rainbow path by property (\ref {eq:property-of-tree}), 
$P(z,x) \cup \{(x,y) \}$ is also rainbow. Since $|P(z,y)| \le  |P(z,x)|$, clearly  we have $|\col(P(z,x) \cup \{ (x,y) \} )| > |\col(P(z,y)|$. It follows
that  there is at least one matching color, say $i$ in the cycle $C = P(z,x) \cup \{(x,y) \} \cup P(y,z)$, such that $\lambda_{H_C} (i) = 1$. But
since $C$ is an even subgraph of $G$, $H_C$ is an even augmentation of $H$ by Claim  \ref {cl:even-subgraph}.  Then by Claim \ref {cl:even-subgraph-color}, 
$\lambda_{H_C} (i)$ should be an even number, a contradiction.

\end {proof}

Now we describe how to construct $T_{i+1}$ from $T_i$ by adding a new level $L_{i+1}$.
For a vertex $x \in L_i$ define $N'(x) = \{ y \in N_{G'}(x) :  \col(x,y) \cap  \col(P(r,x) ) = \emptyset \}$.  
Observe that $N'(x) \cap V(T_i) = \emptyset$: This follows from Claim  \ref {cl:property-of-property}, since if there is an edge in $G'$ from a vertex $x \in L_i$ to
a vertex $y \in L_j$, $j \le i$ then $\col(x,y) \cap \col(P(r,x) )\ne \emptyset$, and therefore $y \notin N'(x)$. 
Define $L_{i+1} = \cup_{x \in L_i} N'(x) $.  Clearly $L_{i+1} \cap V(T_i) = \emptyset$. 
If $L_{i+1} \ne \emptyset$, define  a bipartite graph  $B_i = (L_i, L_{i+1} )$  such that for  $x \in L_{i}$ and $y \in L_{i+1}$, $(x,y) \in E(B_i)$
if and only if $y \in N'(x)$.  Now for each $y \in L_{i+1}$, select one  vertex $y'$  from $L_i$ such that  $(y,y') \in E(B_i)$ to be its parent thus  obtaining the new tree $T_{i+1}$.
From the way we defined $N'(x)$ for $x \in L_i$,  it is clear that property (\ref {eq:property-of-tree}) is satisfied by $T_{i+1}$. 
If $|L_{i+1}| \ge 2 |L_i|$, we proceed to add the next  level.  Otherwise we stop the procedure and define the final tree $T$ to be $T_{i+1}$.

Let $L_\ell$ be the last level added to the tree. Clearly $\ell >  2$.  We observe that $\ell \le \log n + 1 $. Otherwise $|L_{\ell-1}| \ge 2^{t-1} > n$, since $|L_{i+1}| \ge 2 |L_i|$ for $i \le \ell-2$.
Now consider the bipartite graph $B_{\ell-1}$:  For each vertex $x \in L_{\ell-1}$,  we know by applying Claim  \ref {cl:neighborhood-color-bound}  with ${\cal C} = \col (P(r,x))$ that 
$E(x, {\cal C})   \le  4  |\col (P(r,x))| \le 8 \log n $.  But $|E_{G'}(x)| = \deg_{G'}  (x) \ge C' \log n$ and therefore  
$|N'(x)| \ge  (C'-8) \log n $.  Therefore  $B_{\ell-1}$ has at least  $|L_{\ell-1}|  (C'-8) \log n$ edges. Since $|L_\ell| < 2 |L_{\ell-1}|$, there exists a vertex $w \in L_\ell$
such that its degree in $B_{\ell-1}$ is at least $ (C'/2 - 4)  \log n$. Again by applying Claim  \ref {cl:neighborhood-color-bound}, this time
with ${\cal C} = \col (P(r,w))$, at most $ 8 (\log n + 1)$ of these edges can have a common color with  
any edge in  $P(r,w)$.  It follows that if $C$ is taken large enough, $w$ has a neighbor  $w' \in L_{\ell-1}$ such that $\col(w,w') \cap \col(P(r,w) ) = \emptyset$.  This contradicts
Claim \ref {cl:property-of-property} applied to the tree $T=T_\ell$. The situation is depicted in \Cref{fig:proof}. Thus we infer that
$k^2 \le  C  n \log n$, which in turn implies that $n = \Omega\left(\frac {k^2} {\log k } \right)$.  

\begin{figure}[t]
	\begin{center}
		\includegraphics[width=2.5in]{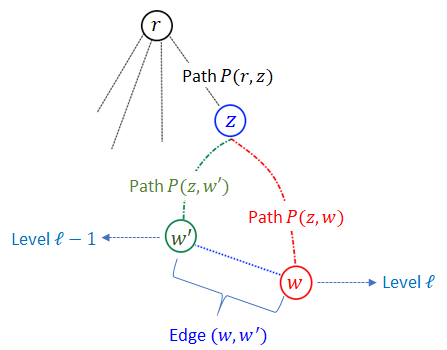}
	\end{center}
	\caption{The cycle formed by the concatenation of $P(z,w'), P(z,w),$ and $(w,w')$  corresponds to an even subgraph in $H$ with an odd number of edges having a particular color.}
	\label{fig:proof}
\end{figure}

\section*{Acknowledgments}
AB thanks Sivakanth Gopi, Nikhil Srivastava, and Luca Trevisan for many useful discussions about this problem. The authors would also like to thank the anonymous reviewers for their useful comments and suggestions.
\bibliographystyle{alpha}
\bibliography{main}

\end{document}